\documentclass[journal]{IEEEtran}

\usepackage{amsmath}
\usepackage{amssymb}
\usepackage{algorithm}
\usepackage{amsthm}
\usepackage{graphicx}
\usepackage{subfig}
\usepackage{enumerate}
\usepackage[noadjust]{cite}
\usepackage{commath}
\usepackage{float}
\usepackage{multicol}
\usepackage{color}
\usepackage{multirow}
\usepackage{array}
\newtheorem{thm}{Theorem}[section]

\newcommand{\aaa}{\mathbf{a}}

\newcommand{\x}{\mathbf{x}}
\newcommand{\y}{\mathbf{y}}

\newcommand{\h}{\mathbf{h}}

\newcommand{\pp}{\mathbf{p}}

\newcommand{\vv}{\mathbf{v}}

\newcommand{\AAA}{\mathbf{A}}

\newcommand{\HH}{\mathbf{H}}
\newcommand{\JJ}{\mathbf{J}}
\newcommand{\SSSS}{\mathbf{S}}

\newcommand{\W}{\mathbf{W}}
\newcommand{\X}{\mathbf{X}}

\newcommand{\rank}{ \textrm{rank} }
\newcommand{\for}{\quad \textrm{for} \quad}

\newcommand{\trace}{ \textrm{trace} }

\begin{document}


\title{Reconstruction of signals from their autocorrelation and cross-correlation vectors, with applications to phase retrieval and blind channel estimation}
\author{
\begin{tabular}[t]{c@{\extracolsep{10em}}c} 
Kishore Jaganathan & Babak Hassibi
\end{tabular}
\\ \vspace{0.2cm}
Department of Electrical Engineering, California Institute of Technology, Pasadena.
\thanks{K. Jaganathan and B. Hassibi were supported in part by the National Science Foundation under grants CCF-0729203, CNS-0932428 and CIF-1018927, by the Office of Naval Research under the MURI grant N00014-08-1-0747, and by a grant from Qualcomm Inc.}  
}

\maketitle

\begin{abstract}
We consider the problem of reconstructing two signals from the autocorrelation and cross-correlation measurements. This inverse problem is a fundamental one in signal processing, and arises in many applications, including phase retrieval and blind channel estimation. In a typical phase retrieval setup, only the autocorrelation measurements are obtainable. We show that, when the measurements are obtained using three simple ``masks", phase retrieval reduces to the aforementioned reconstruction problem.

The classic solution to this problem is based on finding common factors between the $z$-transforms of the autocorrelation and cross-correlation vectors. This solution has enjoyed limited practical success, mainly due to the fact that it is not sufficiently stable in the noisy setting. In this work, inspired by the success of convex programming in provably and stably solving various quadratic constrained problems, we develop a semidefinite programming-based algorithm and provide theoretical guarantees. In particular, we show that almost all signals can be uniquely recovered by this algorithm (up to a global phase). Comparative numerical studies demonstrate that the proposed method significantly outperforms the classic method in the noisy setting.
\end{abstract}


\section{Introduction}

\subsection{Problem Setup}

For the sake of exposition, we begin by considering the discretized $1D$ setting\footnote[2]{The results developed in this work are also applicable to discretized $2D$ signals, we refer the readers to Section \ref{sec:2d} for details.}. Suppose $\x_1 = ( x_1[0], x_1[1] ,\cdots , x_1[L_1-1])^T$ and $\x_2 = ( x_2[0], x_2[1] ,\cdots , x_2[L_2-1])^T$ are the two complex signals of interest. Let $\aaa_1 = ( a_1[1-L_1] , \cdots , a_1[ 0 ] , \cdots , a_1[ L_1 - 1 ] )^T$ and $\aaa_2=( a_2[1-L_2] , \cdots , a_2[ 0 ] , \cdots , a_2[ L_2 - 1 ] )^T$ denote the autocorrelation vectors of $\x_1$ and $\x_2$ respectively, defined as 
\begin{align}
a_1[m] &= \sum_{n=0}^{L_1-1} x_1[n]x_1^\star[n-m], \\
a_2[m] &= \sum_{n=0}^{L_2-1} x_2[n]x_2^\star[n-m],  \nonumber
\end{align}
where, for notational convenience, $x_1[n]$ and $x_2[n]$ have a value of zero outside the intervals $0 \leq n \leq L_1-1$ and $0 \leq n \leq L_2-1$ respectively. Similarly, let  $\aaa_{12}=( a_{12}[1-L_2] , \cdots , a_{12}[ 0 ] , \cdots , a_{12}[ L_1 - 1 ] )^T$ and $\aaa_{21} = ( a_{21}[1-L_1] , \cdots , a_{21}[ 0 ] , \cdots , a_{21}[ L_2 - 1 ] )^T$  denote the cross-correlation vectors of $\x_1$ and $\x_2$, defined as
\begin{align}
a_{12}[m] &= \sum_{n=0}^{L_1-1} x_1[n]x_2^\star[n-m], \\
a_{21}[m] &= \sum_{n=0}^{L_2-1} x_2[n]x_1^\star[n-m]. \nonumber 
\end{align}
Our goal is to uniquely, stably and efficiently reconstruct $\x_1$ and $\x_2$ from $\aaa_1$, $\aaa_2$, $\aaa_{12}$ and $\aaa_{21}$.


\subsection{Trivial Ambiguities}

Observe that the operations of global phase-change and time-shift on $\x_1$ and $\x_2$ do not affect their autocorrelation and cross-correlation vectors. In particular, the autocorrelation vectors of the signals $e^{i\phi}\x_1$ and $e^{i\phi}\x_2$ are $\aaa_1$ and $\aaa_2$ respectively, and their cross-correlation vectors are $\aaa_{12}$ and $\aaa_{21}$. Similarly, the autocorrelation vectors of the signals $\x_1$ and $\x_2$ time-shifted by $c$ units are $\aaa_1$ and $\aaa_2$ respectively, and their cross-correlation vectors are $\aaa_{12}$ and $\aaa_{21}$. Indeed, the assumption that $\x_1$ and $\x_2$ have non-zero values only within the indices $0 \leq n \leq L_1-1$ and $0 \leq n \leq L_2-1$ respectively resolves the time-shift ambiguity when $x_1[0] \neq 0, x_1[L_1-1]\neq 0$ or $x_2[0] \neq 0, x_2[L_2-1]\neq 0$ or $x_1[0] \neq 0, x_2[L_2-1]\neq 0$ or $x_2[0] \neq 0, x_1[L_1-1]\neq 0$.

Consequently, from the autocorrelation and cross-correlation vectors, recovery is in general possible only up to a global-phase and time-shift. These ambiguities are commonly referred to as {\em trivial ambiguities} in literature. Throughout this work, when we refer to successful recovery, it is assumed to be up to the trivial ambiguities.

\subsection{Classic Method}
The classic approach to this reconstruction problem is based on finding common factors between the $z$-transforms of the autocorrelation and cross-correlation vectors. Let $X_1(z)$, $X_2(z)$, $A_1(z)$, $A_2(z)$, $A_{12}(z)$ and $A_{21}(z)$ denote the $z$-transforms of $\x_1$, $\x_2$, $\aaa_1$, $\aaa_2$, $\aaa_{12}$ and $\aaa_{21}$ respectively. The objective is equivalent to reconstruction of the polynomials $X_1(z)$ and $X_2(z)$ from the polynomials $A_1(z)$, $A_2(z)$, $A_{12}(z)$ and $A_{21}(z)$. 

The aforementioned polynomials are related as follows:
\begin{align}
A_1(z) &= X_1(z)X_1^\star(z^{-\star}), \\
A_2(z) &= X_2(z)X_2^\star(z^{-\star}), \nonumber \\
A_{12}(z) &= X_1(z)X_2^\star(z^{-\star}), \nonumber \\
A_{21}(z) &= X_2(z)X_1^\star(z^{-\star}). \nonumber
\end{align}

The key idea is the following: Suppose the polynomials $z^{L_1-1}X_1(z)$ and $z^{L_2-1}X_2(z)$ are co-prime, i.e., they do not have any common roots. Then, $X_1(z)$ can be reconstructed by identifying the common factors between the polynomials $z^{L_1-1}A_1(z)$ and $z^{L_1-1}A_{12}(z)$. Similarly, $X_2(z)$ can be reconstructed by identifying the common factors between the polynomials $z^{L_2-1}A_2(z)$ and $z^{L_2-1}A_{21}(z)$\footnote[3]{The multiplying terms $z^{L_1-1}$ and $z^{L_2-1}$ ensure that the polynomials consist of only non-negative powers of $z$.}. 

In fact, in the classic paper \cite{tong1}, the authors show that the co-prime condition is a necessary and sufficient criterion for successful recovery. Additionally, the authors also provide an algorithm based on finding the {\em greatest common divisor} and {\em residuals} of two polynomials using Sylvester matrices \cite{sylv1}. Numerical simulations show that the algorithm is somewhat stable in the noisy setting.  

For a brief discussion on Sylvester matrices and their use in finding the greatest common divisor and residuals of two  polynomials, we refer the readers to Appendix \ref{sec:sylvestermatrices}.

\subsection{Contributions}

In this work, we develop a semidefinite programming (SDP)-based algorithm. We show that almost all signals can be successfully recovered by this algorithm, subject to the aforementioned co-prime condition (Theorem \ref{thm:main}). In the noisy setting, we conduct extensive numerical simulations and verify the efficacy of the proposed algorithm.

The rest of the paper is organized as follows: In Section 2, we discuss the practical applications of the reconstruction problem. In Section 3, we present our algorithm and provide theoretical guarantees. The results of the various numerical studies are provided in Section 4, and Section 5 concludes the paper.

\section{Motivation}

In this section, we describe two major applications of the reconstruction problem: phase retrieval and blind channel estimation.

\subsection{Phase Retrieval}
\label{sec:pr}

\begin{figure*}
\begin{center}
\subfloat[{Mask \#1}]{\includegraphics[scale=0.4]{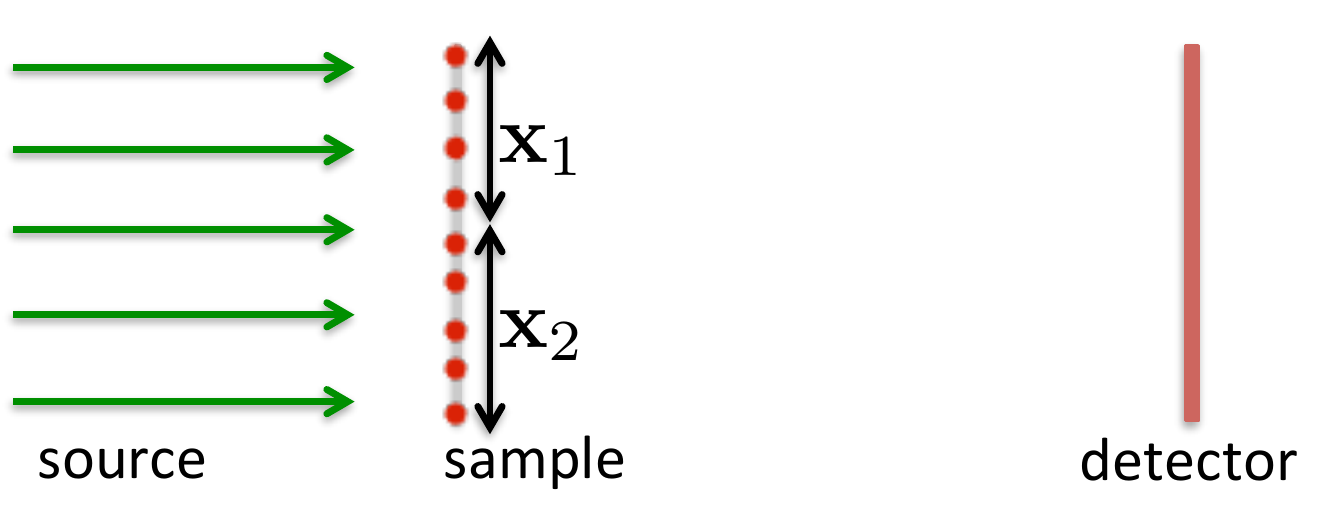}} \hspace{0.5cm}
\subfloat[{Mask \#2}]{\includegraphics[scale=0.4]{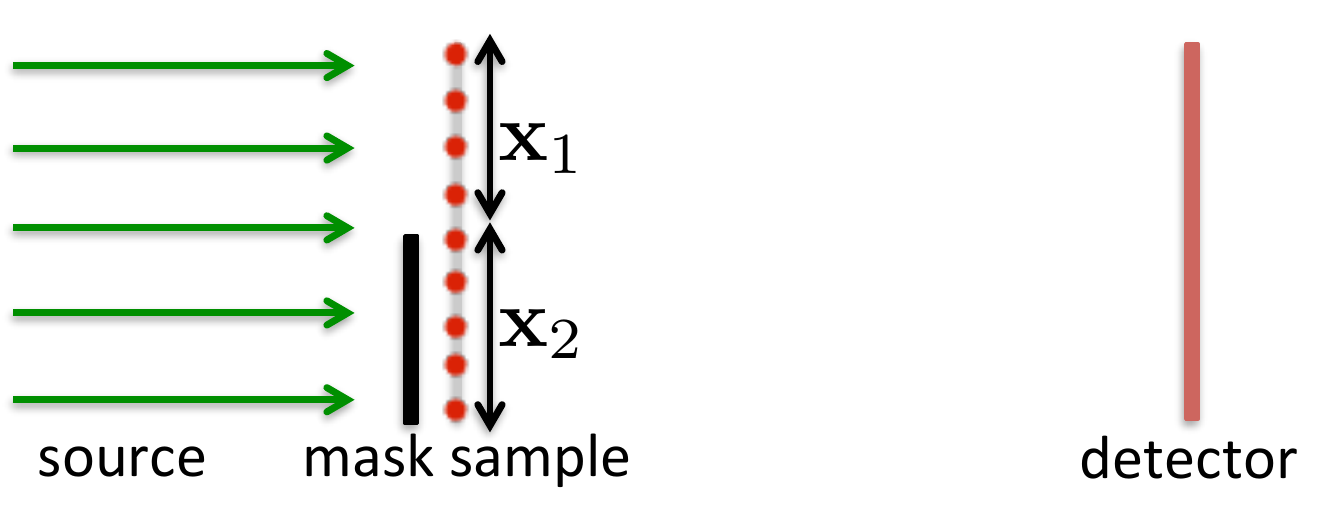}} \hspace{0.5cm}
\subfloat[{Mask \#3}]{\includegraphics[scale=0.4]{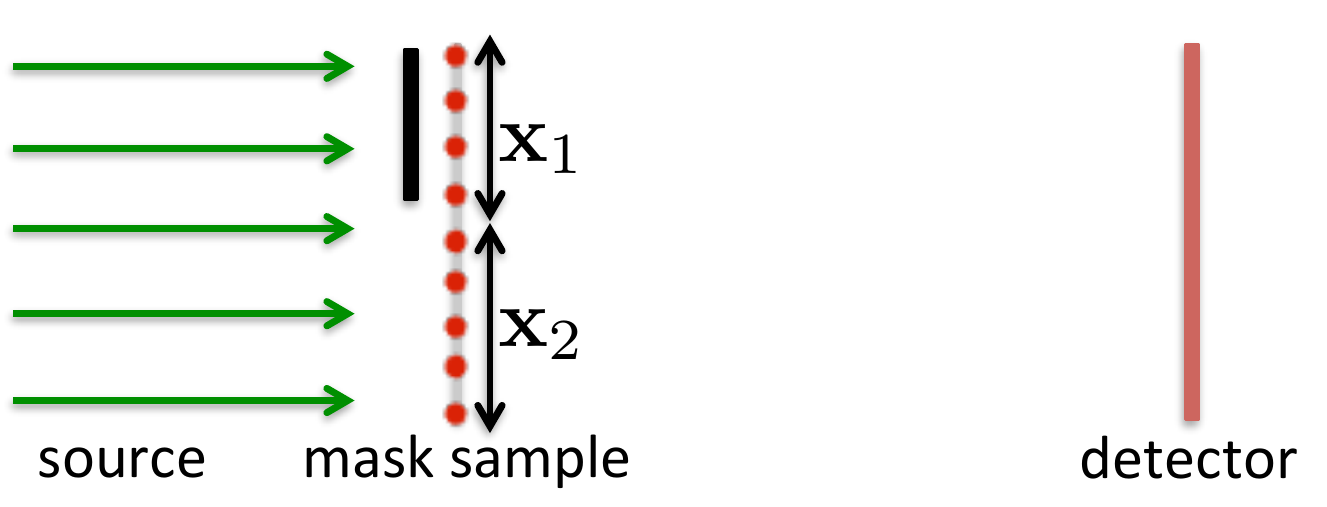}}
\end{center}
\caption{A pictorial representation of a typical 1D phase retrieval setting using the proposed set of masks. A monochromatic beam is incident on the masked sample, and the detector measures the autocorrelation vector of the part of the sample that is not blocked by the mask.}
\label{fig:pr1d}
\end{figure*}
\begin{figure*}
\begin{center}
\subfloat[{Mask \#1}]{\includegraphics[scale=0.4]{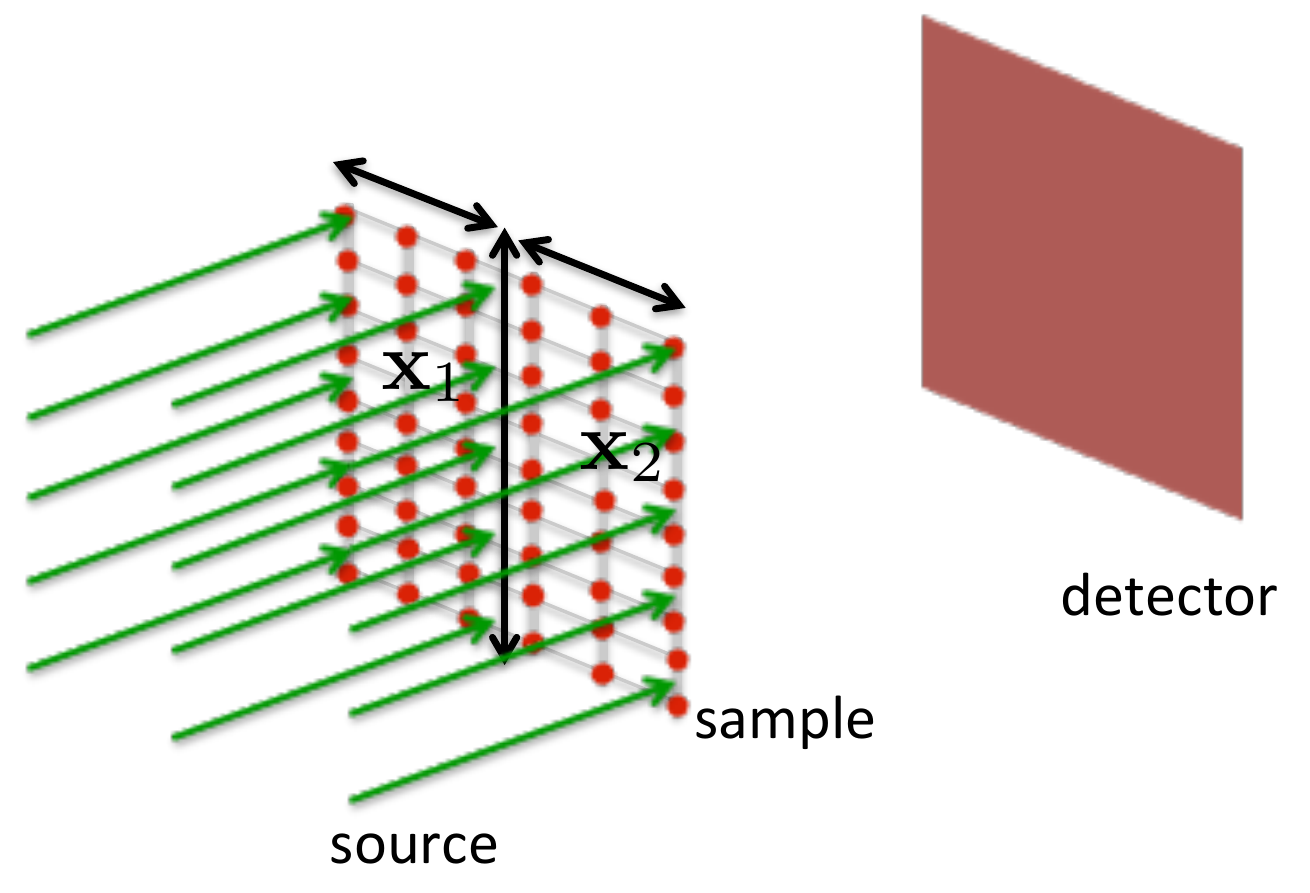}} \hspace{0.5cm}
\subfloat[{Mask \#2}]{\includegraphics[scale=0.4]{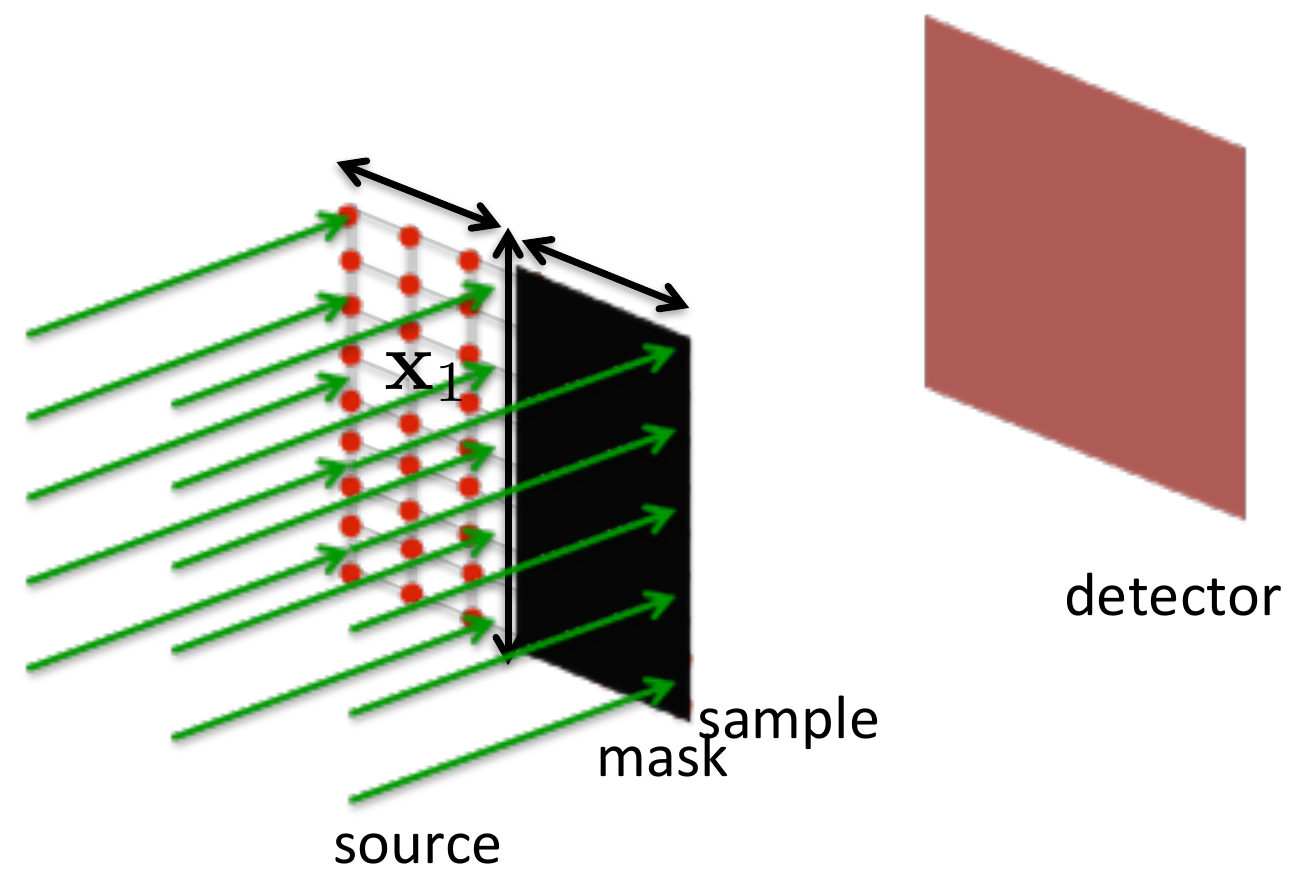}} \hspace{0.5cm}
\subfloat[{Mask \#3}]{\includegraphics[scale=0.4]{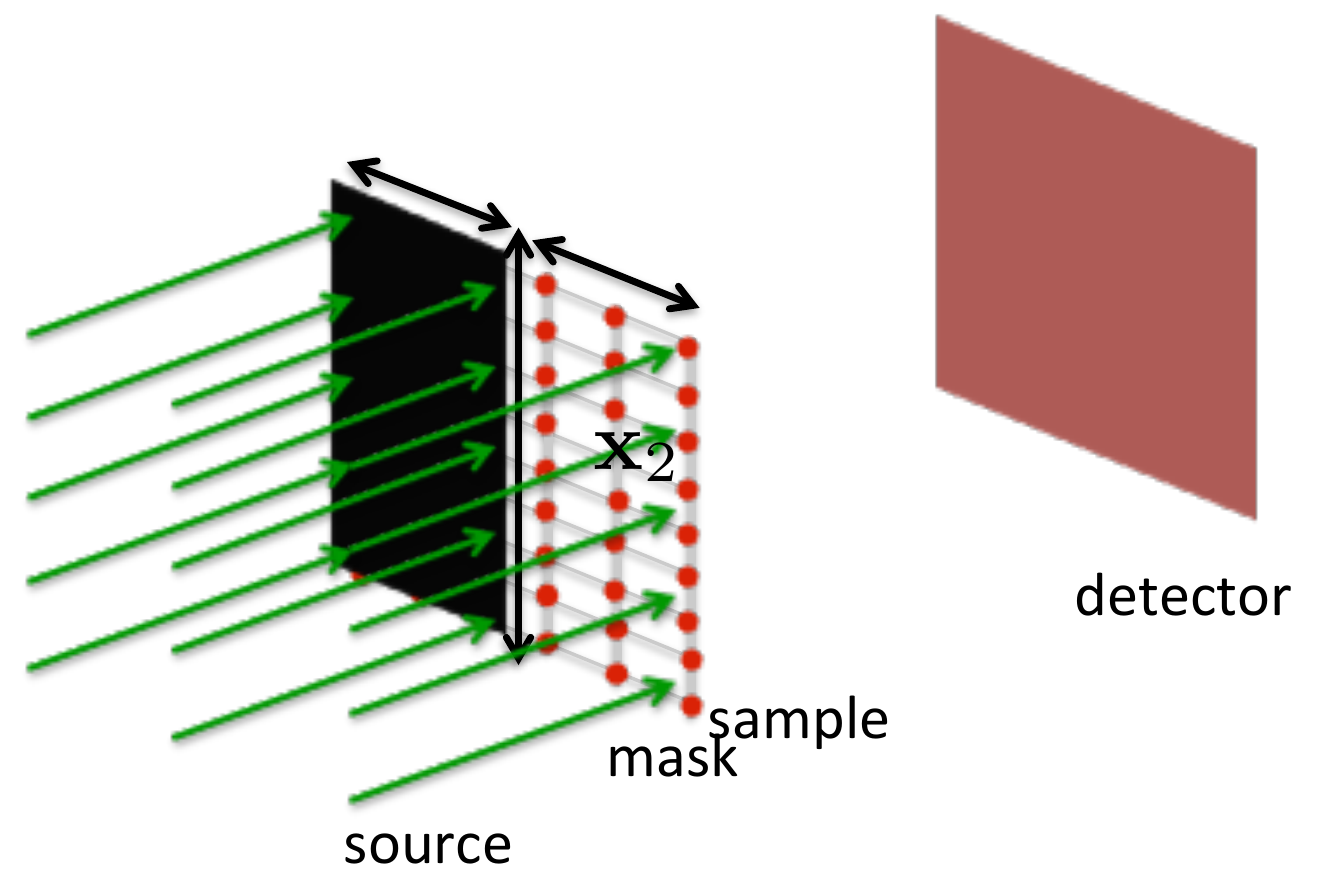}}
\end{center}
\caption{A pictorial representation of a typical 2D phase retrieval setting using the proposed set of masks.}
\label{fig:pr2d}
\end{figure*}

In many practical measurement systems, the measurable quantity is the autocorrelation vector of the signal. Recovering the underlying signal from the autocorrelation measurements is known as phase retrieval.  Phase retrieval arises in many areas of engineering and applied physics, including X-ray crystallography \cite{patt2}, optics \cite{walther, millane}, astronomical imaging \cite{dainty}, bioinformatics \cite{stef} and more. 

Despite an enormous amount of research for nearly hundred years, there are no known efficient and stable algorithms with theoretical guarantees. It is widely accepted that phase retrieval is a computationally difficult problem. We refer the interested readers to \cite{fienup,bauschke} for classic surveys and to \cite{eldarmagazine, kishorereview} for contemporary reviews.

In order to overcome the computational issues of phase retrieval, a common approach in practice is to obtain additional information on the signal by introducing simple modifications to the measurement process. To this end, masking is a popular technique, in which parts of the signal are physically blocked using a {\em mask} and the autocorrelation vector of the rest of the signal is measured \cite{roarke1,roarke2,waller1,waller2}. The premise, in a nutshell, is to introduce redundancy in the reconstruction problem by collecting multiple autocorrelation measurements. In the following, we describe three simple masks and show that, when autocorrelation measurements are obtained using them, phase retrieval is equivalent to the problem of recovering two signals from the autocorrelation and cross-correlation measurements.

Let $\x=(x[0], x[1], \cdots, x[N-1])^T$ be the underlying signal which we wish to determine, and $X(z)$ be its $z$-transform. We use the notation $\x_1 = (x[0], x[1], \cdots , x[L-1])^T$ and $\x_2=(x[L], x[L+1], \cdots , x[N-1])^T$, where $L$ is an integer in the interval $1 \leq L \leq N-2$. In other words, $\x = \begin{bmatrix}\x_1 \\ \x_2\end{bmatrix}$, where $\x_1$ is the signal constructed using the first $L$ entries of $\x$ and $\x_2$ is the signal constructed using the remaining entries of $\x$.

Suppose autocorrelation measurements are collected using the following three masks: 
\begin{itemize}
	\item[(a)] The first mask does not block any part of the signal.

	\item[(b)] The second mask blocks the signal in the interval $L \leq n \leq N-1$. 

	\item[(c)] The third mask blocks the signal in the interval $0 \leq n \leq L-1$. 
\end{itemize}

A pictorial representation is provided in Fig. \ref{fig:pr1d}. Note that the measurements provide the knowledge of the autocorrelation vectors of $\x$, $\x_1$ and $\x_2$. Since we have the relationship
\begin{equation}
X(z) = X_1(z) + z^{-L}X_2(z), \nonumber
\end{equation}
the polynomials $\left(X_1(z) + z^{-L}X_2(z) \right) \left( X_1^\star(z^{-\star}) + z^{L}X_2^\star(z^{-\star})\right)$,  $X_1(z)X_1^\star(z^{-\star})$ and $X_2(z)X_2^\star(z^{-\star})$ are provided by the measurements. Hence, we can infer the polynomial $z^{-L}X_2(z)X_1^\star(z^{-\star}) + z^{L}X_1(z) X_2^\star(z^{-\star})$ from the measurements. Since $z^{-L}X_2(z)X_1^\star(z^{-\star})$ has  terms consisting of only negative powers of $z$ and $z^{L}X_1(z) X_2^\star(z^{-\star})$ has terms consisting of only positive powers of $z$, we can infer the polynomials $X_2(z)X_1^\star(z^{-\star})$ and $X_1(z) X_2^\star(z^{-\star})$ from the measurements.

Therefore, by collecting autocorrelation measurements using the aforementioned three masks, the autocorrelation and cross-correlation vectors of $\x_1$ and $\x_2$ can be inferred. Consequently, phase retrieval reduces to the problem of reconstruction of $\x_1$ and $\x_2$ from their autocorrelation and cross-correlation vectors.

{\bf Remarks}: (i) The total number of phaseless Fourier measurements provided by these masks is $4N$: In order to obtain the autocorrelation vector of a signal of length $N$, it is well-known that $2N$ phaseless Fourier measurements are necessary and sufficient (see Appendix of \cite{kishorestft} for example). The three masks obtain the autocorrelation vectors of signals of lengths $N$, $L$ and $N-L$. The $4N$ quantity has been of significant interest to the phase retrieval community \cite{balan1,balan2,bandeira,ohlsson}.

(ii) In \cite{vpr1,vpr2}, the authors propose a framework called vectorial phase retrieval (VPR). Mathematically, the framework proposed in this section is equivalent to VPR. Indeed, VPR is another framework where the reconstruction problem arises. We refer the interested readers to \cite{vpr1,vpr2} for details.

\subsection{Blind Channel Estimation}

In many communication systems, channel estimation is required in order to be able to achieve reliable communication. A common way of doing this is by periodically sending training sequences known both to the transmitter and receiver \cite{babaktraining}. In scenarios where this is not possible, blind channel estimation is a popular technique, in which the transmitted signal is inferred from the received signal using only the statistical properties of the transmitted signal \cite{sato,godard,tong2}. 

Let $\x$ be a zero-mean and unit-variance i.i.d. random process. Suppose it is transmitted through two linear time-invariant FIR channels $\h_1$ and $\h_2$, or equivalently $H_1(z)$ and $H_2(z)$ in the $z$-transform domain, to obtain random processes $\y_1$ and $\y_2$ respectively. The power spectral densities of $\y_1$ and $\y_2$, denoted by $S_{y_1}(z)$ and $S_{y_2}(z)$, are given by
\begin{align}
S_{y_1}(z) &= H_1(z)H_1^\star(z^{-\star}),  \\
S_{y_2}(z) &= H_2(z)H_2^\star(z^{-\star}), \nonumber
\end{align}
and their cross-spectral densities, denoted by $S_{y_1y_2}(z)$ and $S_{y_2y_1}(z)$, are given by 
\begin{align}
S_{y_1y_2}(z) &= H_1(z)H_2^\star(z^{-\star}), \\
S_{y_2y_1}(z) &= H_2(z)H_1^\star(z^{-\star}). \nonumber
\end{align}
Therefore, the aforementioned measurements provide the knowledge of the autocorrelation and cross-correlation vectors of $\h_1$ and $\h_2$. Consequently, blind channel estimation reduces to the problem of reconstruction of two signals from their autocorrelation and cross-correlation vectors. 

\begin{figure}
\begin{center}
\includegraphics[scale=0.4]{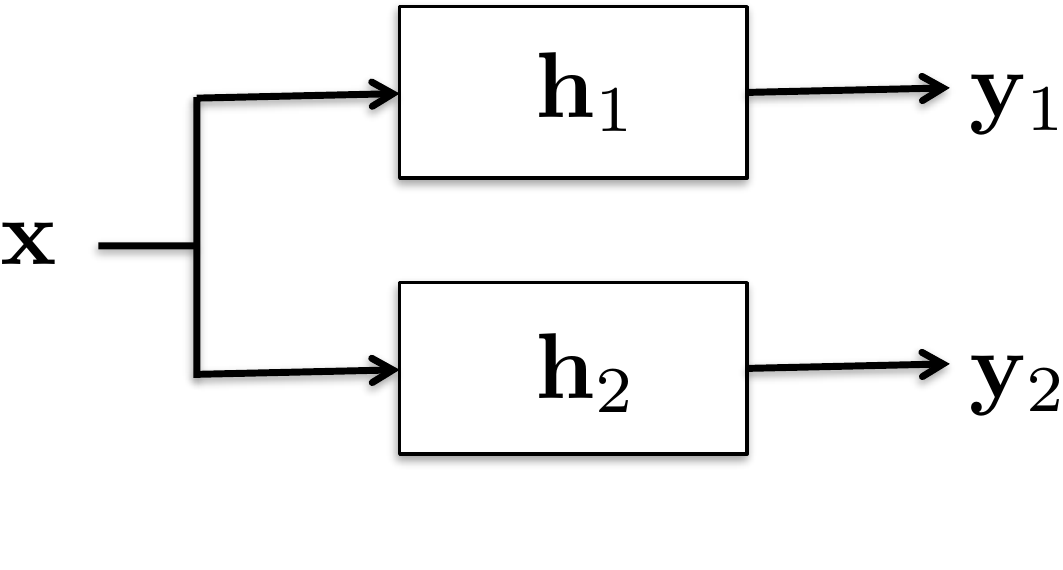}
\end{center}
\caption{The goal of blind channel estimation is to identify $\h_1$ and $\h_2$ using $\y_1$, $\y_2$ and only the statistical properties of $\x$.}
\end{figure}

{\bf Remark}: In \cite{tong3}, the authors show that, if the sampling rate at the receiver is twice the transmission rate (also known as baud rate), then a single linear time-invariant FIR channel mathematically decomposes into two linear time invariant FIR channels. The key idea is the following: The channel $H(z)$ is expressed as 
\begin{equation}
H(z) = H_e(z^2) + z^{-1}H_o(z^2), \nonumber
\end{equation}
where $H_e(z)$ and $H_o(z)$ are the channels involving only the taps corresponding to the even and odd time-slots respectively. Since transmission happens only at even time-slots, the received vector corresponding to the even time-slots is as if the transmitted signal was passed through $H_e(z)$, and the received vector corresponding to the odd time-slots is as if it was passed through $H_o(z)$, thereby converting a single linear time-invariant FIR channel into two linear time-invariant FIR channels. This extends the applicability of the reconstruction problem to scenarios where multiple channels are not available.  

\section{SDP-based reconstruction}

In this section, we first develop the SDP-based algorithm for $1D$ signals and provide theoretical guarantees. Then, we extend the algorithm and theory to $2D$ signals.

Note that the autocorrelation and cross-correlation measurements are quadratic in nature. SDP-based algorithms have been shown to yield robust solutions with theoretical guarantees to various quadratic-constrained optimization problems (see \cite{lovasz, goemans, candespl, candespr, mahdi, li, samet , kishorespr, kishorestft, kishorethesis,eldarcvx,kishorecvx,tropp,romberg} and references therein). Therefore, it is natural to try SDP techniques to solve this problem. An SDP formulation of the reconstruction problem can be obtained by a procedure popularly known as {\em lifting}: 

Let $\x=\begin{bmatrix}\x_1 \\ \x_2\end{bmatrix}$ be the $(L_1+L_2) \times 1$ vector obtained by stacking $\x_1$ and $\x_2$. We embed $\x$ in a higher-dimensional space using the transformation $\X=\x\x^\star$. Since the autocorrelation and cross-correlation measurements are linear in the matrix $\X$, the reconstruction problem reduces to finding a rank-one positive semidefinite matrix which satisfies particular affine constraints. In other words, the reconstruction problem can be equivalently written as
\begin{align}
& \textrm{find} \hspace{1.35cm} \X, \\
& \textrm{subject to} \hspace{0.5cm}\trace(\AAA_m\X) = b_m \for 0 \leq m < M, \nonumber \\
& \hspace{1.9cm} \X \succcurlyeq 0 ~~ \& ~~ \rank(\X) = 1, \nonumber
\end{align}
for appropriate choices of {\em sensing} matrices and measurements $\AAA_m$ and $b_m$, for $0 \leq m < M$, respectively. For example, consider the setup with $L_1=2$ and $L_2=2$. We have $M=12$, as there are $3 + 3$ autocorrelation terms and $3 + 3$ cross-correlation terms. The sensing matrices are
\begin{align*}
&\begin{bmatrix}0 & 0 & 0 & 0 \\ 1 & 0 & 0 & 0 \\ 0 & 0 & 0 & 0 \\ 0 & 0 & 0 & 0\end{bmatrix},
\begin{bmatrix}1 & 0 & 0 & 0 \\ 0 & 1 & 0 & 0 \\ 0 & 0 & 0 & 0 \\ 0 & 0 & 0 & 0\end{bmatrix},
\begin{bmatrix}0 & 1 & 0 & 0 \\ 0 & 0 & 0 & 0 \\ 0 & 0 & 0 & 0 \\ 0 & 0 & 0 & 0\end{bmatrix} ,\\
&\begin{bmatrix}0 & 0 & 0 & 0 \\ 0 & 0 & 0 & 0 \\ 0 & 0 & 0 & 0 \\ 0 & 0 & 1 & 0\end{bmatrix}, 
\begin{bmatrix}0 & 0 & 0 & 0 \\ 0 & 0 & 0 & 0 \\ 0 & 0 & 1 & 0 \\ 0 & 0 & 0 & 1\end{bmatrix},
\begin{bmatrix}0 & 0 & 0 & 0 \\ 0 & 0 & 0 & 0 \\ 0 & 0 & 0 & 1 \\ 0 & 0 & 0 & 0\end{bmatrix},\\
&\begin{bmatrix}0 & 0 & 0 & 0 \\ 0 & 0 & 1 & 0 \\ 0 & 0 & 0 & 0 \\ 0 & 0 & 0 & 0\end{bmatrix},
\begin{bmatrix}0 & 0 & 1 & 0 \\ 0 & 0 & 0 & 1 \\ 0 & 0 & 0 & 0 \\ 0 & 0 & 0 & 0\end{bmatrix},
\begin{bmatrix}0 & 0 & 0 & 1 \\ 0 & 0 & 0 & 0 \\ 0 & 0 & 0 & 0 \\ 0 & 0 & 0 & 0\end{bmatrix},\\
&\begin{bmatrix}0 & 0 & 0 & 0 \\ 0 & 0 & 0 & 0 \\ 0 & 0 & 0 & 0 \\ 1 & 0 & 0 & 0\end{bmatrix},
\begin{bmatrix}0 & 0 & 0 & 0 \\ 0 & 0 & 0 & 0 \\ 1 & 0 & 0 & 0 \\ 0 & 1 & 0 & 0\end{bmatrix},
\begin{bmatrix}0 & 0 & 0 & 0 \\ 0 & 0 & 0 & 0 \\ 0 & 1 & 0 & 0 \\ 0 & 0 & 0 & 0\end{bmatrix},\\
\end{align*}
and the corresponding measurements are $a_1[-1], a_1[0], a_1[1]$, $a_2[-1], a_2[0], a_2[1]$, $a_{12}[-1], a_{12}[0], a_{12}[1]$, $a_{21}[-1], a_{21}[0]$ and $a_{21}[1]$.

To obtain an SDP formulation, one possibility is to relax the rank constraint, resulting in the following convex algorithm: 
\begin{algorithm}[H]
\caption{SDP-based reconstruction algorithm}
{\bf Inputs}: The autocorrelation and cross-correlation measurements $b_m$ for $0 \leq m < M$, the signal lengths $L_1$ and $L_2$. \\
{\bf Outputs}: Signal estimates $\hat{\x}_1$ and $\hat{\x}_2$. \\
\begin{itemize}
\item Obtain the $(L_1+L_2) \times (L_1+L_2)$ matrix $\hat{\X}$ by solving
\begin{align}
\label{eqn:sdp}
& \textrm{find} \hspace{1.35cm} \X, \\
& \textrm{subject to} \hspace{0.5cm}\trace(\AAA_m\X) = b_m \for 0 \leq m < M, \nonumber \\
& \hspace{1.9cm} \X \succcurlyeq 0. \nonumber
\end{align}
\item Calculate the best rank-one approximation of $\hat{\X}$ through SVD, and get $\hat{\x}\hat{\x}^\star$.
\item Return $\hat{\x}_1=(\hat{x}[ 0 ], \hat{x}[ 1 ], \cdots, \hat{x}[ L_1 - 1 ])^T$ and $\hat{\x}_2=(\hat{x}[ L_1 ], \hat{x}[ L_1 + 1 ], \cdots, \hat{x}[ L_1 + L_2 - 1 ])^T$.
\end{itemize}
\label{algo:sdp}
\end{algorithm}
We provide the following theoretical guarantee for recovery using Algorithm \ref{algo:sdp}:
\begin{thm}
Suppose the signals $\x_1$ and $\x_2$, of lengths $L_1$ and $L_2$ respectively, are such that the polynomials $z^{L_1-1}X_1(z)$ and $z^{L_2-1}X_2(z)$ are co-prime, and $x_1[0],x_2[0]\neq0$. For almost all such $\x_1$ and $\x_2$, the convex program (\ref{eqn:sdp}) has a unique feasible point, namely, $\begin{bmatrix}\x_1 \\ \x_2\end{bmatrix}\begin{bmatrix}\x_1^\star ~ \x_2^\star\end{bmatrix}$, and thus the outputs of Algorithm \ref{algo:sdp} are $\hat{\x}_1=\x_1$ and $\hat{\x}_2=\x_2$. 
\label{thm:main}
\end{thm}

\begin{proof}
The proof of this theorem involves dual certificates and Sylvester matrices. An overview of the method of dual certificates is provided in Appendix \ref{sec:dualcertificates}, and relevant properties of Sylvester matrices are described in Appendix \ref{sec:sylvestermatrices}. 

As before, we use the notations $\x = \begin{bmatrix}\x_1 \\ \x_2 \end{bmatrix}$, $N = L_1 + L_2$ and $L=L_1$ for the sake of simplicity. Let $T_{\x}$ denote the set of Hermitian matrices of the form
\begin{equation}
T_{\x} = \{ \x\h^\star + \h\x^\star : \h \in \mathbb{C}^N\}, \nonumber
\end{equation}
and $T_{\x}^\perp$ be its orthogonal complement. We use $\HH_{T_{\x}}$ and $\HH_{T_{\x}^\perp}$ to denote the projections of a matrix $\HH$ onto the subspaces $T_{\x}$ and $T_{\x}^\perp$ respectively.

By construction, the matrix $\x\x^\star$ is a feasible point of (\ref{eqn:sdp}). Standard duality arguments in semidefinite programming (see Section \ref{sec:dualcertificates} for details) show that the following conditions are sufficient for $\x\x^\star$ to be the unique optimizer, i.e., the unique feasible point, of (\ref{eqn:sdp}): 
\begin{enumerate}[]
\item {\em Condition 1}: There exists a {\em dual certificate} matrix $\W = \sum_{m=0}^{M-1} \left( \lambda_m\AAA_m + \lambda_m^\star\AAA_m^\star\right)$, where $\lambda_0, \lambda_1, \cdots , \lambda_{M-1}$ are scalar complex numbers, with the following properties:
	\begin{itemize}
	\item[(a)] $\W \succcurlyeq 0$,
	\item[(b)] $\W \x = 0$,
	\item[(c)] $\rank( \W ) = N-1$.
	\end{itemize}
\item {\em Condition 2}: If $\HH \in T_{\x}$ and $\trace(\AAA_m\HH) = 0$ for $0 \leq m < M$, then $\HH = 0$ is the only solution.
\end{enumerate}
In words, the matrix $\W$ is parametrized by scalar variables $\lambda_0, \lambda_1, \cdots , \lambda_{M-1}$ through the aforementioned relationship. The process of dual certificate construction deals with assigning values to $\lambda_0, \lambda_1, \cdots , \lambda_{M-1}$ in such a way that the resulting $\W$ satisfies the properties specified in Condition 1. Condition 2 typically deals with well-known properties of polynomials, and is in general straightforward to show. 

The range space of $\sum_{m=0}^{M-1} \left( \lambda_m\AAA_m + \lambda_m^\star\AAA_m^\star\right)$, parametrized by $\lambda_0, \lambda_1, \cdots , \lambda_{M-1}$, is the set of all Hermitian $N \times N$ matrices which are such that the submatrices corresponding to the  $0 \leq n \leq L-1$ rows and columns,  $L \leq n \leq N-1$ rows and columns, $0 \leq n \leq L-1$ rows and $L \leq n \leq N-1$ columns, and $L \leq n \leq N-1$ rows and $0 \leq n \leq L-1$ columns are Toeplitz matrices.

Let $\SSSS_{z^{L}X_1(z),z^{N-L}X_2(z)}$ be the $N \times N$ Sylvester matrix constructed using the two polynomials $z^{L}X_1(z)$ and $z^{N-L}X_2(z)$, i.e., $\SSSS_{z^{L}X_1(z),z^{N-L}X_2(z)}$ is the following matrix:
{\scriptsize\begin{equation}
\begin{bmatrix}
& x_2[0] & 0 & . & 0 & -x_1[0] & 0 & . & 0 \\
& x_2[1]  & x_2[0] & . &  0 & -x_1[1]  & -x_1[0] & . &  0 \\
& x_2[2] & x_2[1] & . & .  & -x_1[2] & -x_1[1] & . & . \\
& . & x_2[2] & . & .& . & -x_1[2] & . & . \\
& . & . & .  & .& . & . & . & .\\
& . & . & . & x_2[0] & . & . & . & . \\
& 0 & . & . & x_2[1] & 0 & . & . & -x_1[0]\\
& 0 & 0 & . & x_2[2]& 0 & 0 & . & -x_1[1] \\
& . & 0 & .  & .& . & 0 & . & -x_1[2]\\
& . & . & .  & .& . & . & . & .\\
& .  & . & .  & x_2[L_2-1]& .  & . & . & -x_1[L_1-1]\\
& 0 & 0 & .  & 0& 0 & 0 & . & 0\\
\end{bmatrix}. \nonumber
\end{equation}}
The $0 \leq n \leq L-1$ columns of $\SSSS_{z^{L}X_1(z),z^{N-L}X_2(z)}$ are such that the $n$th column is $\x_2$ shifted by $n$ units, and the $L \leq n \leq N-1$ columns are such that the $n$th column is $-\x_1$ shifted by $n-L$ units. We refer the readers to Section \ref{sec:sylvestermatrices} for a description of the intuition behind defining such a matrix.

To show that Condition 1 is satisfied for $\x\x^\star$, we propose the following {\em dual certificate}: 
\begin{equation}
\W = \SSSS_{z^{L}X_1(z),z^{N-L}X_2(z)}^\star \SSSS_{z^{L}X_1(z),z^{N-L}X_2(z)}.
\label{eqn:dualcertificate}
\end{equation}
The matrix $\W$ is clearly in the range space of $\sum_{m=0}^{M-1} \left( \lambda_m\AAA_m + \lambda_m^\star\AAA_m^\star\right)$: Since the first $L$ columns of $\SSSS_{z^{L}X_1(z),z^{N-L}X_2(z)}$ are shifted copies of the $0$th column, their inner products have a Toeplitz structure. The same applies to the inner products between the remaining $N-L$ columns, and the inner products between the first $L$ columns and the remaining $N-L$ columns.

(a) $\W$ is positive semidefinite by construction. 

(b) Since $z^{N-L}X_2(z) \times {z^{L-1}}X_1(z) - z^{L} X_1(z) \times {z^{N-L-1}}X_2(z) = 0$, we have $\SSSS_{z^{L}X_1(z),z^{N-L}X_2(z)} \begin{bmatrix}\x_1 \\ \x_2 \end{bmatrix} = 0$. This is due to a property of Sylvester matrices described in (\ref{eq:sylv1}) and (\ref{eq:sylv2}). Alternately, $\SSSS_{z^{L}X_1(z),z^{N-L}X_2(z)} \begin{bmatrix}\x_1 \\ \x_2 \end{bmatrix} = 0$ can be verified by simply multiplying the quantities. Therefore, we have $\W \x = 0$. 
	
(c) The $x_1[0],x_2[0]\neq0$ condition ensures that the degrees of the polynomials $z^{L}X_1(z)$ and $z^{N-L}X_2(z)$ are $L$ and $N-L$ respectively. The polynomial $z$ is the greatest common divisor of $z^{L}X_1(z)$ and $z^{N-L}X_2(z)$, due to the fact that $z^{L-1}X_1(z)$ and $z^{N-L-1}X_2(z)$ are co-prime. Therefore, the rank of $\SSSS_{z^{L}X_1(z),z^{N-L}X_2(z)}$ is equal to $(L) + (N-L) - (1) = N - 1$. This is due to a property of  Sylvester matrices described in (\ref{eq:sylv0}), which states that the rank of the Sylvester matrix is equal to the sum of the degrees of the two associated polynomials minus the degree of their greatest common divisor. Consequently, we have $\rank(\W) = N-1$. 

Next, we show that Condition 2 is satisfied for almost all $\x\x^\star$. Since $\HH \in T_\x$, we can write $\HH =  \x \h^\star + \h \x^\star$ for some $\h = ( h[ 0 ] , h[ 1 ] , \cdots , h[ N - 1 ] )^T$. Instead of working with the length $N$ complex vector $\h$, we work with the length $2N$ real vector $\begin{bmatrix}\textrm{Re}(\h) \\ \textrm{Im}(\h)\end{bmatrix}$, where the operations $\textrm{Re}(\h)$ and $\textrm{Im}(\h)$ obtain the element-wise real and imaginary parts of $\h$ respectively. In other words, instead of working with the complex variables, we work with the real variables that form their real and imaginary parts.

The equation $\trace(\AAA_m\HH) = 0$, for any $m$, is linear with respect to $\begin{bmatrix}\textrm{Re}(\h) \\ \textrm{Im}(\h)\end{bmatrix}$. For example, the equation in complex variables
\begin{equation}
\nonumber x[0]h^\star[L-1]+h[0]x^\star[L-1] =0
\end{equation}
can be equivalently written as two equations in real variables:
{\scriptsize \begin{equation}
\nonumber \begin{bmatrix}\textrm{Re}(x[L-1])  & \textrm{Re}(x[0]) & \textrm{Im}(x[L-1])  & \textrm{Im}(x[0]) \\
-\textrm{Im}(x[L-1])  & \textrm{Im}(x[0]) & \textrm{Re}(x[L-1])  & -\textrm{Re}(x[0])
 \end{bmatrix}  \begin{bmatrix} \textrm{Re}(h[0]) \\ \textrm{Re}(h[L-1]) \\ \textrm{Im}(h[0]) \\ \textrm{Im}(h[L-1]) \end{bmatrix} = 0.
\end{equation} }

Let $\JJ_\x\begin{bmatrix}\textrm{Re}(\h) \\ \textrm{Im}(\h)\end{bmatrix}=0$ denote the constraints corresponding to the equations $\trace(\AAA_m\HH) = 0$ for $0 \leq m < M$. Note that $\JJ_\x$ is an $M \times 2N$ matrix, where $M=4N-4$, whose entries are either the entries of $\begin{bmatrix}\textrm{Re}(\x) \\ \textrm{Im}(\x)\end{bmatrix}$ with a plus or minus sign, or $0$. Instead of focusing on the precise structure of $\JJ_\x$, we complete the proof using the following property of $\JJ_\x$: The determinant of each $2N-1 \times 2N-1$ submatrix of $\JJ_\x$ is a finite-degree polynomial function of the entries of $\begin{bmatrix} \textrm{Re}(\x) \\ \textrm{Im}(\x) \end{bmatrix}$. 

Finite-degree polynomial functions have the following well-known property: they are either $0$ everywhere, or non-zero almost everywhere. Therefore, the determinant of any particular $2N-1 \times 2N-1$ submatrix of $\JJ_\x$ is either $0$ for all $\x$, or non-zero for almost all $\x$. Consequently, one of the following is true: the determinant of every $2N-1 \times 2N-1$ submatrix of $\JJ_\x$ is $0$ for all $\x$, or there exists at least one $2N-1 \times 2N-1$ submatrix which has a non-zero determinant for almost all $\x$. By substituting $\x = (1, 0 , \cdots, 0 )^T$, we eliminate the possibility of every $2N-1 \times 2N-1$ determinant being $0$ for all $\x$. As a result, the rank of $\JJ_\x$ is at least $2N-1$ for almost all $\x$. 

Furthermore, the vector corresponding to $\h = ic\x$ is in the null space of $\JJ_\x$ for any real constant $c$, due to the fact that the corresponding $\HH = -ic\x\x^\star + ic\x\x^\star$ is $0$. Therefore, for almost all $\x$, the rank of $\JJ_\x$ is equal to $2N-1$, and $\h = ic\x$ for any real constant $c$ is the only feasible solution. In other words, $\HH = -ic\x\x^\star + ic\x\x^\star = 0$ is the only matrix that satisfies both $\HH \in T_{\x}$ and $\trace(\AAA_m\HH) = 0$ for $0 \leq m < M$.
\end{proof}

\subsection{Extension to $2D$ Signals}
\label{sec:2d}

The results developed in this section for $1D$ signals can be extended to $2D$ signals using the following trick:

Suppose $\x_{1,2D}$ and $\x_{2,2D}$ are two $2D$ signals of size $L_{11} \times L_{12}$ and $L_{21} \times L_{22}$ respectively. Let $\aaa_{1,2D}, \aaa_{2,2D}$ and $\aaa_{12,2D}, \aaa_{21,2D}$ be their $2D$ autocorrelation and cross-correlation matrices respectively. Also, let $\x_{1,1D} = \textrm{vec}(\x_{1,2D})$ denote the $1D$ vector constructed by stacking the columns of $\x_{1,2D}$. The $1D$ autocorrelation vector of $\x_{1,1D}$, denoted by $\aaa_{1,1D}$, can be inferred from $\aaa_{1,2D}$. This can be seen as follows:

For $m \geq 0$, we have
{\scriptsize \begin{align}
\nonumber &\aaa_{1,1D}[m] = \sum_{n=0}^{L_{11}L_{12}-1} x_{1,1D}[n]x_{1,1D}^\star[n-m], \\
\nonumber &= \sum_{l_2=0}^{L_{12}-1}\sum_{l_1=m\textrm{mod}L_{11}}^{L_{11}-1} x_{1,1D}[l_2L_{11}+l_1]x_{1,1D}^\star[l_2L_{11}+l_1- m] \\
\nonumber &+ \sum_{l_2=0}^{L_{12}-1}\sum_{l_1=0}^{m\textrm{mod}L_{11}-1} x_{1,1D}[l_2L_{11}+l_1]x_{1,1D}^\star[l_2L_{11}+l_1- m], \\
& \nonumber = \sum_{l_2=0}^{L_{12}-1}\sum_{l_1=m\textrm{mod}L_{11}}^{L_{11}-1}
x_{1,2D}[l_1,l_2]x_{1,2D}^\star[l_1- m\textrm{mod}L_{11},l_2-\left\lfloor\frac{m}{L_{11}}\right\rfloor] \\
& \nonumber + \sum_{l_2=0}^{L_{12}-1}\sum_{l_1=0}^{m\textrm{mod}L_{11}-1} 
x_{1,2D}[l_1,l_2]x_{1,2D}^\star[l_1- m~\textrm{mod}~L_{11}+L_{11},l_2-\left\lfloor\frac{m}{L_{11}}\right\rfloor-1],\\
& = a_{1,2D}[m~\textrm{mod}~L_{11} , \left\lfloor\frac{m}{L_{11}}\right\rfloor ] + a_{1,2D}[m~\textrm{mod}~L_{11} - L_{11} , \left\lfloor\frac{m}{L_{11}}\right\rfloor + 1  ], \nonumber
\end{align}}
where, for notational convenience, $x_{1,1D}[n]$ has a value of zero outside the interval $0 \leq n \leq L_{11}L_{12}-1$ and $x_{1,2D}[n_1,n_2]$ has a value of zero outside the interval $0 \leq n_2 \leq L_{12}-1$. Since the values of $\aaa_{1,1D}$ for $m < 0$ are the conjugates of the values of $\aaa_{1,1D}$ for $m>0$, $\aaa_{1,1D}$ is completely characterized by $\aaa_{1,2D}$. Similarly, the $1D$ autocorrelation and cross-correlation vectors $\aaa_{2,1D}$,  $\aaa_{12,1D}$ and $\aaa_{21,1D}$ can be inferred from the $2D$ autocorrelation and cross-correlation matrices $\aaa_{2,2D}$,  $\aaa_{12,2D}$ and $\aaa_{21,2D}$ respectively.

In other words, the autocorrelation and cross-correlation vectors of $\x_{1,1D}$ and $\x_{2,1D}$ can be inferred from the $2D$ measurements. Using Theorem \ref{thm:main}, we conclude that almost all signals $\x_{1,1D}$ and $\x_{2,1D}$, which are such that the polynomials $z^{L_{11}L_{12}-1}X_{1,1D}(z)$ and $z^{L_{21}L_{22}-1}X_{2,1D}(z)$ are co-prime, and $x_{1,1D}[0],x_{2,1D}[0]\neq0$, can be uniquely reconstructed by Algorithm \ref{algo:sdp}. Finally, the desired signals $\x_{1,2D}$ and $\x_{2,2D}$ can be recovered from $\x_{1,1D}$ and $\x_{2,1D}$ respectively by appropriate reshaping. 

Consequently, the three masks proposed for phase retrieval in Section \ref{sec:pr} generalizes to the $2D$ setting as follows: Let $\x$ be a $2D$ signal of size $N_1 \times N_2$, and $L$ be an integer in the interval $1 \leq L \leq N_2-2$:
\begin{itemize}
\item[(a)] The first mask does not block any part of the signal.
\item[(b)] The second mask blocks the signal in the columns $L \leq n \leq N_2-1$. 
\item[(c)] The third mask blocks the signal in the columns $0 \leq n \leq L-1$.
\end{itemize}
A pictorial representation of the setup is provided in Fig. \ref{fig:pr2d}. 

{\bf Remarks}: (i) One could also perform the $\textrm{vec}(.)$ operation by stacking rows. 

(ii) The $2D$ autocorrelation and cross-correlation measurements correspond to affine constraints in the lifted domain. As a result, there is no need to calculate the $1D$ autocorrelation and cross-correlation measurements of the vectorized signals while implementing the algorithm in practice. 

(iii) In \cite{eldarnew}, the authors explore the general connection between $1D$ and $2D$ phase retrieval using similar tricks.

\subsection{Noisy setting}

In practice, the measurements are contaminated by additive noise. One way of implementing Algorithm \ref{algo:sdp} in the noisy setting is:
 \begin{align}
 \label{eq:sdpnoisy}
& \textrm{minimize} \hspace{0.9cm} \sum_{m=0}^{M-1}\abs{\trace(\AAA_m\X) - b_m }^2, \\
& \textrm{subject to} \hspace{0.9cm}  \X \succcurlyeq 0, \nonumber
\end{align}
where $b_m$, for $0 \leq m < M$, are the noisy autocorrelation and cross-correlation measurements. We choose $\ell_2$-norm in the objective function keeping in mind the fact that measurement noise is typically AWGN. In settings where the noise vector is known to be sparse, one could choose $\ell_1$-norm instead \cite{l1}. Since the desired solution is a rank one matrix, one could also add a $\trace(\X)$ term to the objective function with an appropriate regularizer \cite{fazel}.

\section{Numerical Simulations}

In this section, we demonstrate the performance of Algorithm \ref{algo:sdp} using numerical simulations.

First, we perform a comparative study of the Sylvester matrix-based and SDP-based algorithms in the noisy setting. The Sylvester matrix-based algorithm proposed in \cite{tong1} is implemented as described in the remark at the end of Appendix \ref{sec:sylvestermatrices}, and the SDP-based algorithm is implemented as described in (\ref{eq:sdpnoisy}).

We perform a total of $50$ trials for $L_1=32,L_2=32$ and $L_1=48,L_2=16$ setups. In each trial, the two signals $\x_1$ and $\x_2$ are sampled uniformly at random from a sphere of radius $\sqrt{L_1}$ and $\sqrt{L_2}$ respectively. If the signals do not satisfy $\abs{x_1[0]},\abs{x_2[0]} \geq 0.2$, then they are sampled again. Their autocorrelation and cross-correlation vectors are computed, and corrupted with additive zero mean Gaussian noise of appropriate variance (decided by the SNR). 

The normalized mean-squared error (NMSE), defined as 
\begin{equation}
\mathbb{E}\left[ \min_\phi\frac{\|\x - e^{i\phi}\hat{\x}\|_2^2}{\|\x\|_2^2}\right],
\end{equation}
where $\x = \begin{bmatrix}\x_1 \\ \x_2 \end{bmatrix}$, is plotted as a function of SNR in Fig. \ref{fig:sim1}. The approximately linear relationship between the NMSE and SNR in the logarithmic scale indicates that the reconstruction using both methods is stable in the noisy setting. Further, the superior performance of the SDP-based method can be clearly seen. Convex methods are known to be very robust to noise in general. So, this observation is along the expected lines. 

\begin{figure}
\begin{center}
{\includegraphics[scale=0.45]{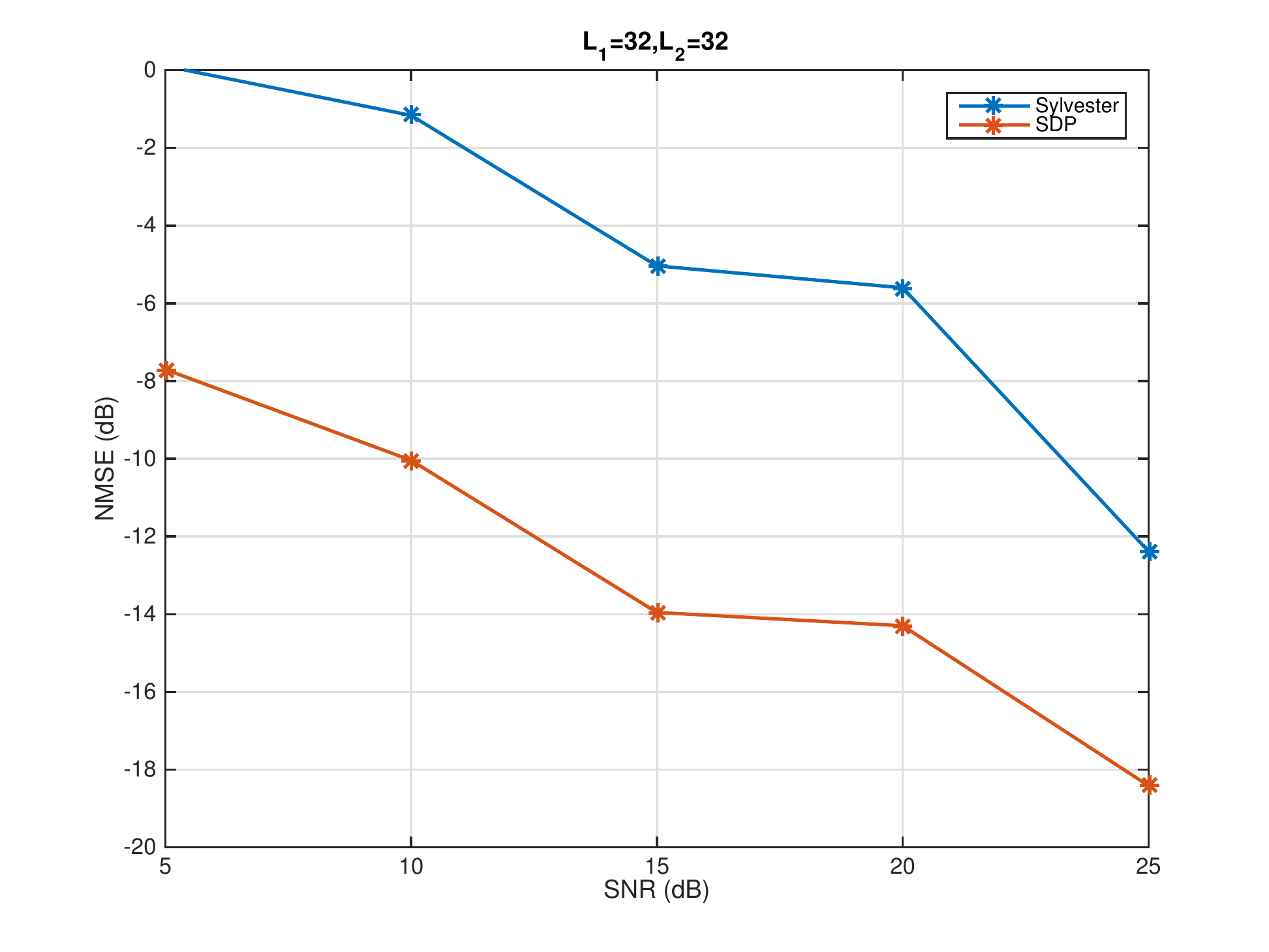}}
{\includegraphics[scale=0.45]{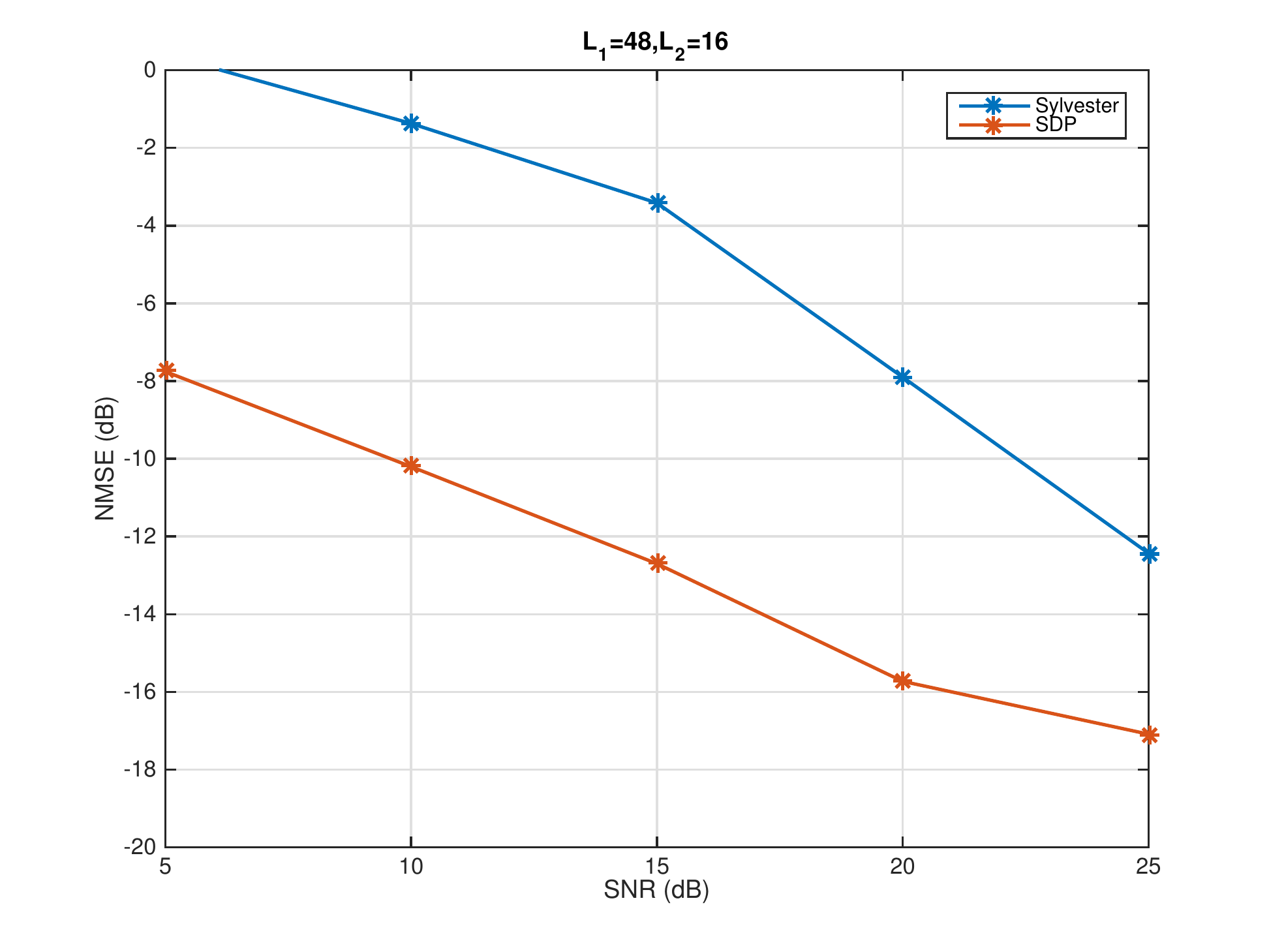}}
\end{center}
\caption{A comparative study of the NMSE vs SNR for the Sylvester matrix-based and SDP-based algorithms.}
\label{fig:sim1}
\end{figure}

Next, we demonstrate another important feature of the SDP-based framework. In applications like phase retrieval, one could potentially collect additional measurements using more masks. In such setups, the Sylvester matrix-based framework cannot make use of the additional measurements. In contrast, the additional measurements can be added as extra affine constraints in the SDP-based framework.  

Consider the setup with $N=64$ and $L=32$. While the setup is similar to $L_1=32,L_2=32$, there is a small difference in the way the noise is modeled. As described in Section \ref{sec:pr}, the cross-correlation vectors are not directly measured and instead calculated using three autocorrelation measurements, because of which their variance is three times higher. 

The signal $\x$ is sampled as before. Fig. \ref{fig:sim2} compares the stability of the SDP-based method in the following two setups: (1) no additional measurements are considered and (2) additional measurements using masks defined by $L=16,48$ are considered. As expected, the plot suggests that the additional measurements lead to a further improvement in stability.

\begin{figure}
\begin{center}
{\includegraphics[scale=0.45]{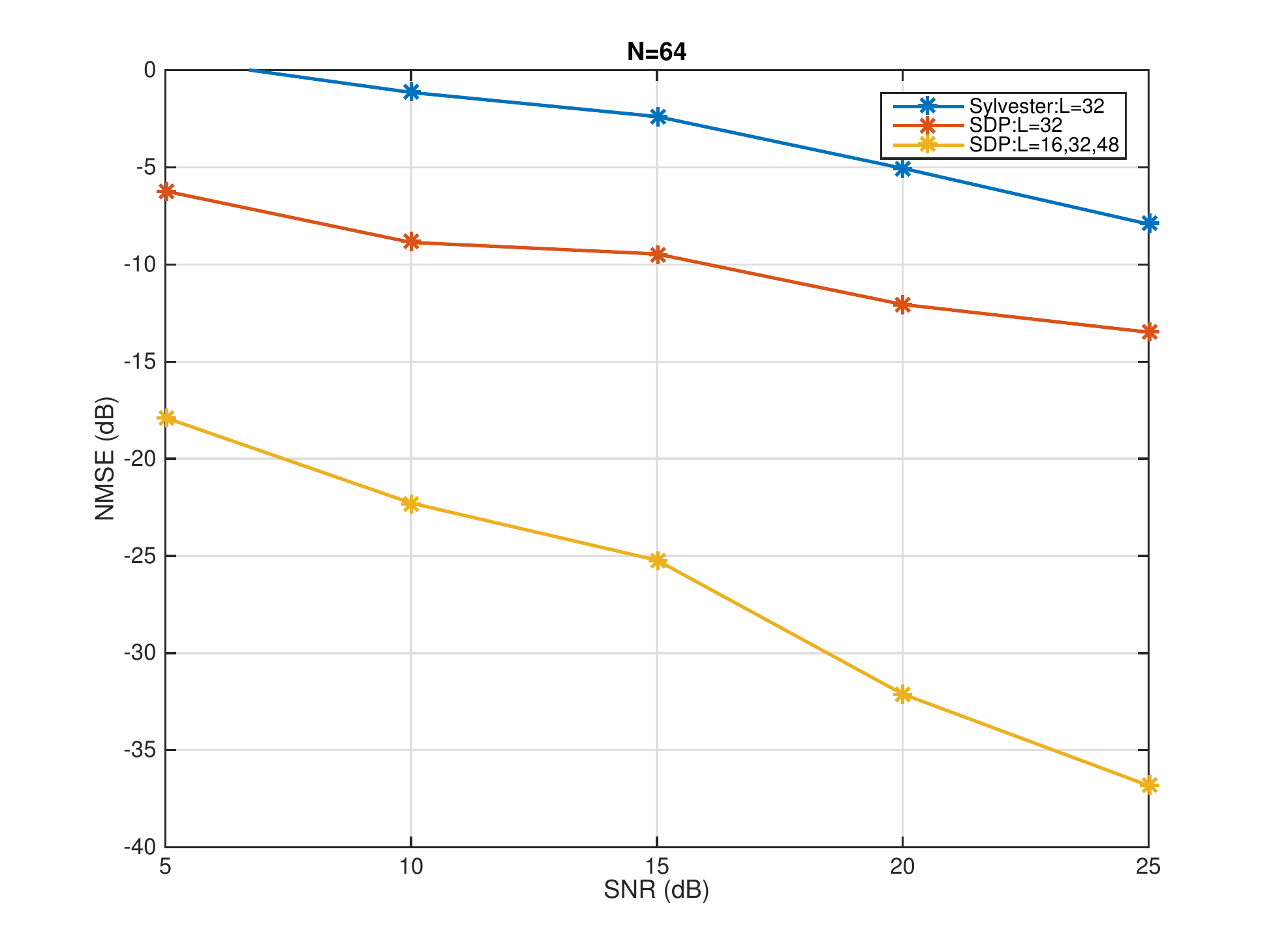}}
\end{center}
\caption{NMSE vs SNR for the SDP-based algorithm when additional measurements are available.}
\label{fig:sim2}
\end{figure}

\section{Conclusions}

In this work, we considered the problem of reconstruction of signals from their autocorrelation and cross-correlation measurements. We first described two applications where this reconstruction problem naturally arises: phase retrieval and blind channel estimation. In the phase retrieval setup, where only the autocorrelation vectors can be measured, we proposed three simple masks and showed that phase retrieval is equivalent to the aforementioned reconstruction problem when measurements are obtained using them.

Then, we formulated this problem as a convex program using the standard lifting method and provided theoretical guarantees. In particular, we showed that the convex program uniquely identifies almost all signals  in the noiseless setting. In the noisy setting, we demonstrated the superior stability of this approach over the standard Sylvester matrix-based approach through numerical simulations.

\section{Method of Dual Certificates}
\label{sec:dualcertificates}

In this section, we provide an overview of the method of dual certificates. This technique is applicable to a wide class of optimization problems. Here, we focus our attention on using it as a theoretical tool to analyze feasibility-type SDPs.  

Consider the following primal optimization problem:
\begin{align}
\label{eqn:primal}
& \textrm{find} \hspace{1.3cm} \X, \\
& \textrm{subject to} \hspace{0.5cm}\trace(\AAA_m\X) = b_m \for 0 \leq m < M, \nonumber \\
& \hspace{1.9cm} \X \succcurlyeq 0, \nonumber
\end{align}
where $\X$ is an $N \times N$ Hermitian matrix. The objective is to derive a set of {\em tractable} conditions which ensure that the matrix $\x\x^\star$ is the unique feasible point, i.e., the unique optimizer, of (\ref{eqn:primal}). 
The dual optimization problem is given by  
\begin{align}
& \max_{\lambda_0,\lambda_1,\cdots,\lambda_{M-1}} \hspace{0cm} -\sum_{m=0}^{M-1} \left(\lambda_m b_m + \lambda_m^\star b^\star_m \right), \\
& \textrm{subject to} \hspace{0.55cm}\sum_{m=0}^{M-1} \left( \lambda_m\AAA_m + \lambda_m^\star\AAA_m^\star \right) \succcurlyeq 0. \nonumber
\end{align}
We use the definition 
\begin{equation}
\W =\sum_{m=0}^{M-1} \left( \lambda_m\AAA_m + \lambda_m^\star\AAA_m^\star \right). \nonumber
\end{equation}
The matrix $\W$, which is parametrized by the dual variables $\lambda_0, \lambda_1, \cdots ,\lambda_{M-1}$, is commonly referred to as {\em dual certificate} in literature.
 
 KKT conditions show that, for $\x\x^\star$ and $\lambda_0, \lambda_1, \cdots ,\lambda_{M-1}$ to be the primal and dual optimizers respectively\footnote[2]{Since the primal optimization problem is a feasibility problem (i.e., there is no objective function), every feasible point is an optimizer. In order to obtain the dual optimization problem, a constant $0$ can be used as the objective function.}, the following criteria are necessary and sufficient: 
\begin{itemize}
	\item  $\trace(\AAA_m\x\x^\star) = b_m \for 0 \leq m < M$ (primal feasibility), 
	\item $\W \succcurlyeq 0$ (dual feasibility, Condition 1a),
	\item $\trace( \W \x\x^\star) = 0$ (complementary slackness).
\end{itemize}

The complementary slackness criterion can be equivalently written as $\W \x = 0$ (Condition 1b) due to the fact that when $\W \succcurlyeq 0$,  $\trace( \W\x\x^\star ) = 0$ and $\W\x = 0$ are equivalent statements.

Next, the goal is to ensure that the matrix $\x\x^\star$ is the only primal optimizer. Suppose $\x\x^\star + \HH$ is a primal optimizer. In what follows, we derive tractable conditions which are only satisfied by $\HH=0$.

Let $T_{\x}$ denote the set of Hermitian matrices of the form
\begin{equation}
T_{\x} = \{ \x\h^\star + \h\x^\star : \h \in \mathbb{C}^N\}, \nonumber
\end{equation}
and $T_{\x}^\perp$ be its orthogonal complement. The set $T_{\x}$ can be interpreted as the tangent space at $\x\x^\star$ to the manifold of Hermitian matrices of rank one. We use $\HH_{T_{\x}}$ and $\HH_{T_{\x}^\perp}$ to denote the projections of the matrix $\HH$ onto the subspaces $T_{\x}$ and $T_{\x}^\perp$ respectively. The matrix $\HH$ is such that
\begin{equation}
\trace( \AAA_m\HH ) = \trace( \AAA_m^\star\HH ) = 0 \for 0 \leq m < M \nonumber
\end{equation}
and $\HH_{T_\x^\perp} \succcurlyeq 0$ (primal feasibility). The first constraint is due to the fact that  $\trace(\AAA_m\x\x^\star) = \trace( \AAA_m ( \x\x^\star + \HH ) ) = b_m$ and $\trace(\AAA_m^\star\x\x^\star) = \trace( \AAA_m^\star ( \x\x^\star + \HH ) ) = b_m^\star$ for $0 \leq m < M$. The second constraint is due to the following: $(\HH_{T_\x^\perp})\x = 0$ and $\x^\star(\HH_{T_\x^\perp}) = 0$ by construction, and for any vector perpendicular to $\x$, say $\x^\perp$, we have
\begin{align}
\x\x^\star + \HH \succcurlyeq 0 &\Rightarrow  \x\x^\star + \HH_{T_\x} + \HH_{T_\x^\perp} \succcurlyeq 0 \nonumber \\
& \Rightarrow \x^{\perp\star} (\x\x^\star + \HH_{T_\x} + \HH_{T_\x^\perp} )\x^\perp \geq 0 \nonumber \\
& \Rightarrow \x^{\perp\star} (\HH_{T_\x^\perp} )\x^\perp \geq 0. \nonumber
\end{align}

As a consequence of the first constraint, we have $\trace( \W \HH )=0$ regardless of the choice of $\lambda_0, \lambda_1, \cdots ,\lambda_{M-1}$. Note that 
\begin{equation}
0 = \trace( \W \HH ) = \trace( \W \HH_{T_\x} ) + \trace( \W \HH_{T_\x^\perp} ).\nonumber
\end{equation} 
The condition $\W\x = 0$ ensures that $\trace( \W \HH_{T_\x} ) = 0$, because of which we have $\trace( \W \HH_{T_\x^\perp} ) = 0$. Since $\W$ and $\HH_{T_\x^\perp}$ are both positive semidefinite matrices, if $\W\x=0$ and $\rank( \W ) = N-1$ (Condition 1c), then $\HH_{T_\x^\perp}=0$ is the only possibility.  

We have shown that, if Conditions 1a, 1b and 1c are satisfied, then any primal optimizer must be of the form $\x\x^\star + \HH_{T_\x}$. In other words, Conditions 1a, 1b and 1c restrict the matrix $\HH$ to the set $T_\x$. Finally, suppose $\HH=0$ is the only matrix that satisfies both $\HH \in T_{\x}$ and $\trace(\AAA_m\HH) = 0$ for $0 \leq m < M$ (Condition 2). Then, the matrix $\x\x^\star$ is the only optimizer of (\ref{eqn:primal}).

Therefore, if Conditions 1a, 1b, 1c and Condition 2 are satisfied, then $\x\x^\star$ is the unique optimizer of (\ref{eqn:primal}). Indeed, since (\ref{eqn:primal}) is a feasibility problem, the conditions ensure that $\x\x^\star$ is its unique feasible point. 

\section{Sylvester Matrices}
\label{sec:sylvestermatrices}

Sylvester matrices are typically encountered when one is interested in common factors between two univariate polynomials. In particular, let $P_1(z) = Q(z)R_1(z)$ and $P_2(z) = Q(z)R_2(z)$ be two polynomials such that $R_1(z)$ and $R_2(z)$ are co-prime, i.e., do not have any common factors.  Given $P_1(z)$ and $P_2(z)$, the goal is to identify their {\em greatest common divisor} $Q(z)$, and their {\em residuals} $R_1(z)$ and $R_2(z)$. 

Suppose $P_1(z) = p_{1,0}z^{d_{p_1}} + p_{1,1}z^{d_{p_1}-1} + \cdots + p_{1,d_{p_1}}$ and $P_2(z) = p_{2,0}z^{d_{p_2}} + p_{2,1}z^{d_{p_2}-1} + \cdots + p_{2,d_{p_2}}$, and $\pp_1=(p_{1,0}, p_{1,1}, \cdots , p_{1,d_{p_1}})^T$ and $\pp_2=(p_{2,0}, p_{2,1}, \cdots , p_{2,d_{p_2}})^T$ are the corresponding coefficient vectors. Then, the Sylvester matrix associated with $P_1(z)$ and $P_2(z)$, denoted by $\SSSS_{P_1(z), P_2(z)}$, is the following $(d_{p_1} + d_{p_2}) \times (d_{p_1} + d_{p_2})$ matrix:
{\scriptsize\begin{equation}
\begin{bmatrix} 
& p_{2,0} & 0 & . & 0 & -p_{1,0} & 0 & . & 0 \\
& p_{2,1}  & p_{2,0} & . &  0 & -p_{1,1} & -p_{1,0} & . &  0 \\
& .  & p_{2,1} & . & .  & . & -p_{1,1} & . & . \\
& . & . & . & .& . & . & . & .\\
& . & . & . & . & -p_{1,d_{p_1}} & . & . & . \\
& p_{2,d_{p_2}} & . & . & . & 0 & -p_{1,d_{p_1}} & . & . \\
& 0 & p_{2,d_{p_2}} & . & . & 0 & 0 & . & .\\
& 0 & 0 & . & .& . & 0 & . & .\\
& . & 0 & . & .& . & . & . & .\\
& . & . & . & p_{2,d_{p_2}-1} & . & . & . & -p_{1,d_{p_1}-1}\\
& 0  & 0 & . & p_{2,d_{p_2}} & 0  & 0 &  . & -p_{1,d_{p_1}}\\
\end{bmatrix}.
\label{sylvestermatrix}
\end{equation}}
The first $d_{p_1}$ columns are shifted copies of $\pp_2$ and the remaining $d_{p_2}$ columns are shifted copies of $-\pp_1$. 

The rank of the Sylvester matrix is a function of the degrees of the two associated polynomials and their greatest common divisor. In particular, the following holds \cite{sylv2}:
\begin{equation}
\label{eq:sylv0}
\rank(\SSSS_{P_1(z), P_2(z)}) = d_{p_1} + d_{p_2} - d_{q},
\end{equation}
where $d_{q}$ is the degree of $Q(z)$. Consequently, $\SSSS_{P_1(z), P_2(z)}$ has full rank iff the polynomials $P_1(z)$ and $P_2(z)$ do not have any common factors. 

Furthermore, the null space of the Sylvester matrix provides information about the residuals of the associated polynomials. In particular, let $V_1(z) = v_{1,0}z^{d_{p_1}-1} + v_{1,1}z^{d_{p_1}-2} + \cdots + v_{1,d_{p_1}-1}$ and $V_2(z) =v_{2,0}z^{d_{p_2}-1} + v_{2,1}z^{d_{p_2}-2} + \cdots + v_{2,d_{p_2}-1} $, and $\vv_1=(v_{1,0}, v_{1,1}, \cdots , v_{1,d_{p_1}-1})^T$ and $\vv_2=(v_{2,0}, v_{2,1}, \cdots , v_{2,d_{p_2}-1})^T$ be the corresponding coefficient vectors. The vector $\begin{bmatrix} \vv_1 \\ \vv_2 \end{bmatrix}$ belongs to the null space of $\SSSS_{P_1(z), P_2(z)}$, i.e.,  
\begin{equation}
\SSSS_{P_1(z), P_2(z)} \begin{bmatrix} \vv_1 \\ \vv_2 \end{bmatrix} = 0
\label{eq:sylv1}
\end{equation}
iff
\begin{equation}
P_2(z)V_1(z) - P_1(z)V_2(z) = 0. 
\label{eq:sylv2}
\end{equation}
The proof of this is straightforward: The constraint that the coefficients of every power of $z$ in (\ref{eq:sylv2}) must be $0$ results in the same set of equations as (\ref{eq:sylv1}). In fact, this is precisely the idea behind the structure of Sylvester matrices. Consequently, if $v_{1,0},v_{1,1}, \cdots , v_{1,d_q-2}$ and $v_{2,0},v_{2,1}, \cdots , v_{2,d_q-2}$ are set to $0$, i.e., the degrees of the residuals are forced to be at most $d_{p_1}-d_q$ and $d_{p_2}-d_q$ respectively, then the only solution to (\ref{eq:sylv2}) is $V_1(z)=R_1(z)$ and $V_2(z) = R_2(z)$ up to a constant factor. 

The left null space of the Sylvester matrix contains information about the greatest common divisor of the associated polynomials. The details are beyond the scope of this paper, and can be found in \cite{sylv2}.

{\bf Remark}: When $P_1(z) = X_1^\star(z^{-\star}) \times z^{L_1-1}X_1(z) = z^{L_1-1}A_1(z) $ and $P_2(z) = X_1^\star(z^{-\star}) \times z^{L_2-1}X_2(z) = z^{L_2-1}A_{21}(z) $, and the degrees of the residuals are forced to be at most $L_1-1$ and $L_2-1$ respectively, the only solution to (\ref{eq:sylv1}) is $\begin{bmatrix} \x_1 \\ \x_2 \end{bmatrix}$ up to a constant factor if $z^{L_1-1}X_1(z)$ and $z^{L_2-1}X_2(z)$ are co-prime, and $x_1[0],x_2[0] \neq 0$  (to resolve the time-shift ambiguity). This is the Sylvester matrix-based solution proposed in \cite{tong1}. In the noisy setting, the Sylvester matrix is constructed using the noisy measurements, and the right singular vector corresponding to the smallest singular value is returned as the estimate.

\end{document}